\documentclass[aps,showpacs,amssymb,amsfonts,superscriptaddress,pra,twocolumn,nofootinbib,10pt]{revtex4-1}
\usepackage{bm,bbm}
\usepackage{times}
\usepackage{graphicx,amsthm,epstopdf,amsmath}
\usepackage{mwe}
\usepackage{subfigure,color}
\usepackage[draft]{fixme}
\usepackage{changes,color}
\usepackage[super]{nth}
\usepackage[top=1.25in, bottom=1.25in, left=0.70in, right=0.70in]{geometry}
\definecolor{myurlcolor}{rgb}{0,0,0.7}
\definecolor{myrefcolor}{rgb}{0.8,0,0}
\usepackage{hyperref}
\hypersetup{colorlinks, linkcolor=myrefcolor,
citecolor=myurlcolor, urlcolor=myurlcolor}
\usepackage{hyperref}
\usepackage{mathtools}
\usepackage{xcolor}
\usepackage{booktabs}

\newcommand{\beq}[0]{\begin{equation}}
\newcommand{\eeq}[0]{\end{equation}}
\newcommand{\bw}[0]{\begin{widetext}}
\newcommand{\ew}[0]{\end{widetext}}
\newcommand{\bc}[0]{\begin{center}}
\newcommand{\ec}[0]{\end{center}}
\newcommand{\bwn}[0]{\begin{widetext}\begin{eqnarray}}
\newcommand{\ewn}[0]{\end{eqnarray}\end{widetext}}
\newcommand{\beqn}[0]{\begin{eqnarray}}
\newcommand{\eeqn}[0]{\end{eqnarray}}

\newcommand{\np}[0]{{\it e.g.}}

\newcommand{\gesy}[0]{genuinely entangled subspaces\;}

\newcommand{\tzn}[0]{{\it i.e.}}

\newcommand{\uroj}[0]{\mathrm{i}}

\newcommand{\ket}[1]{|#1\rangle}
\newcommand{\bra}[1]{\langle #1 |}

\newcommand{\jedynka}[0]{\mathbbm{1}}

\newcommand{\non}[0]{\nonumber\\}

\newcommand{\spann}[0]{\mathrm{span}}

\newtheorem{theorem}{Theorem}

\def\calB{{\cal B}}
\def\calC{{\cal C}}

\def\calG{{\cal G}}
\def\calH{{\cal H}}

\def\calM{{\cal M}}

\def\calP{{\cal P}}

\def\calV{{\cal V}}

\def\frakL{{\frak L}}

\newtheorem{thm}{Theorem}

\newtheorem{lem}[thm]{Lemma}
\newtheorem{defi}[thm]{Definition}

\newtheorem{fakt}[thm]{Fact}

\newtheorem{uwaga}[thm]{Remark}
\newtheorem{obserwacja}[thm]{Observation}

\newcommand{\beu}{\begin{equation}}
\newcommand{\eeu}{\end{equation}}

\newcommand{\be}{\begin{eqnarray}}
\newcommand{\ee}{\end{eqnarray}}

\newcommand{\ba}{\begin{array}}
\newcommand{\ea}{\end{array}}

\newcommand{\ce}{\mathbb{C}}
\newcommand{\cee}[1]{\mathbb{C}^{#1}}

\selectcolormodel{gray}
\begin{document}
\title{From  unextendible product bases to genuinely entangled subspaces }
\author{Maciej Demianowicz}
\affiliation{{\it\small Atomic and Optical Physics Division, Department of Atomic, Molecular and Optical Physics, Faculty of Applied Physics and Mathematics,
Gda\'nsk University of Technology, Narutowicza 11/12, 80–-233 Gda\'nsk, Poland}}
\author{Remigiusz Augusiak}
\affiliation{{\it\small Center for Theoretical Physics, Polish Academy of Sciences, Aleja Lotnik\'ow 32/46, 02-668 Warsaw, Poland}}

\begin{abstract}
Unextendible product bases (UPBs) are interesting mathematical objects arising in composite Hilbert spaces that have found various applications in quantum information theory, for instance in a construction of bound entangled states or Bell inequalities without quantum violation.
They are closely related to another important notion, completely entangled subspaces (CESs), which are those that do not contain any fully separable pure state. Among CESs one finds a class of subspaces in which all vectors are not only entangled, but are genuinely entangled.
Here we explore the connection between UPBs and such genuinely entangled subspaces (GESs)
and provide classes of nonorthogonal UPBs that lead to GESs for any number of parties and local dimensions. 
We then show how these subspaces can be immediately utilized for a simple general construction of genuinely entangled states in any such multipartite scenario.
\end{abstract}
\maketitle

\section{Introduction} 

Entangled states 
play a central role in virtually any information processing protocol in quantum networks, for example quantum teleportation or quantum key distribution (see, \np, \cite{Shor-algorithm,teleportation,state-merging,experimental-QEC}). They are also vital for nonlocality and steering -- other valuable resources in quantum information theory \cite{przegladowka-nonlocality,przegladowka-steering}. First considered in bipartite setups, entanglement has been quickly recognized to be a particularly powerful supply when shared among several parties. Of the rich variety of types of entanglement in such setups it is its genuine multiparty manifestation which appears to be the most useful in practice, as for instance in quantum metrology \cite{GezaMetro2012,HyllusMetro2012,AugusiakMetro2016}. In recent years, we have thus witnessed an unrelenting interest in the literature in such states both from the theoretical (see, \np, Ref. \cite{OtfriedTheo2017}) and the experimental (see, \np, Refs. \cite{KlemptExp2014,BrunnerExp2017})  points  of view. 

At the heart of the research on multiparty quantum states lies the problem of the verification whether a state is entangled \cite{Terhal2002,GuhneToth-review}.  In its full generality the problem is known to be extremely difficult \cite{Gurvits-NP,Gharibian-NP} (see also \cite{np-jordik} for recent advances).  From this perspective, construction of states for which some {\it a priori} knowledge about entanglement properties is available is very desirable. One particular approach relies on the construction of completely entangled subspaces (CESs), that is subspaces void of fully product vectors \cite{ces-bhat,ces-partha}. There follows an easy observation that states with support in such subspaces are necessarily entangled, attaining in turn the goal. 
The notion of a CES is intimately connected with the notion of unextendible product bases (UPBs) \cite{upb-bennett,cmp-upb,upb-alon-lovasz,Bravyi,Cohen2008}. The latter are sets of product, possibly mutually non--orthogonal, vectors spanning a proper subset of a given Hilbert space with the property that no other product vector exists in the complement of their span. 
From the very definition of a UPB it follows that the orthogonal complement of a subspace spanned by it is a CES. We can thus attack the problem stated above from a different angle by analyzing a complementary one.  Such approach proved to be very fruitful and resulted in the constructions of entangled states which are positive after the partial transpose \cite{upb-bennett,upb-pittenger}. Notably, UPBs have also found some surprising applications in other areas as they were used to construct Bell inequalities with no quantum violation \cite{bell-no-quantum}.
From this perspective the task of providing means of constructing UPBs becomes particularly important. Most of the efforts in this area have been focused on UPBs with the orthogonality conditions imposed, let us call them orthogonal UPBs (oUPBs), due to their immediate applications mentioned above. 
Despite intensive research \cite{cmp-upb,min-size-upb,johnston-upb-qubit,UPB3x4,UPB4x4}, a fully general construction has not been developed (albeit see \cite{NisetCerf}). At the same time, much less attention  has been devoted to UPBs with the orthogonality condition dropped, so--called non--orthogonal UPBs (nUPBs), and in consequence their applications in quantum information are largely unexplored (see, however, \cite{LeinaasUPB,Skowronek}).

The picture of the relation between product bases and entangled subspaces depicted above is missing an important element.
An apparent weakness of the  approaches so far
to the construction of CESs from UPBs, regardless of the type of the latter,
is that one has not, in principle, any control of the type of entanglement in the arising entangled subspaces. As we discussed earlier, this knowledge is essential in most of the cases as we usually demand the entanglement to be of the genuine multiparty kind. In fact, the already--known UPBs lead to CESs containing biproduct, \tzn, not genuinely entangled, states.  Hence, one is naturally led to the problem of designing UPBs, which by construction would give rise to subspaces  only containing genuinely entangled states.
These subspaces may be called genuinely entangled subspaces (GESs), in analogy to the completely entangled ones. Although not having been explicitely named as above, they seem to have been first considered in \cite{schmidt-rank}, where subspaces with bounded Schmidt rank were analyzed. 
As noted there, a not too large random subspace will typically be genuinely entangled, so the mere existence of GESs is trivially settled. As a matter of fact, such subspaces can be easily constructed from nUPBs. This can be achieved by randomly drawing a properly chosen number of fully product states. This argument was originally presented in Ref. \cite{cmp-upb} in relation to CESs, however, its extension to the case of GESs is straightforward. Although this solves the principal task of a construction of a GES from a UPB, it adds little to an understanding of the mathematical structure of GESs. From this viewpoint, it is desirable to have access to analytical constructions of the latter in the general multiparty case
and to address the problem of the constructions of GESs in full generality in relation to UPBs. This is where our research fits in.

Motivated by the existence of a completely entangled subspace in the orthocomplement of the span of an unextendible product basis, we ask for such bases which by construction guarantee the orthogonal states to be genuinely entangled, or, in other words the resulting CES to be a GES. We turn our attention to nUPBs, which allows us to provide general examples valid for any number of parties holding systems of any local dimensions (see Fig. \ref{podprzestrzenie-rys}), without resorting to arguments about random states and typicality. Importantly, any state supported on a GES is genuinely entangled, and, implementing a well known idea \cite{upb-bennett}, we provide examples of such mixed states in a general multiparty scenario.
Entanglement in these states is particularly easy to be detected and we give entanglement witnesses for them.

\begin{figure}[h!]
\includegraphics[height=4cm,width=7cm]{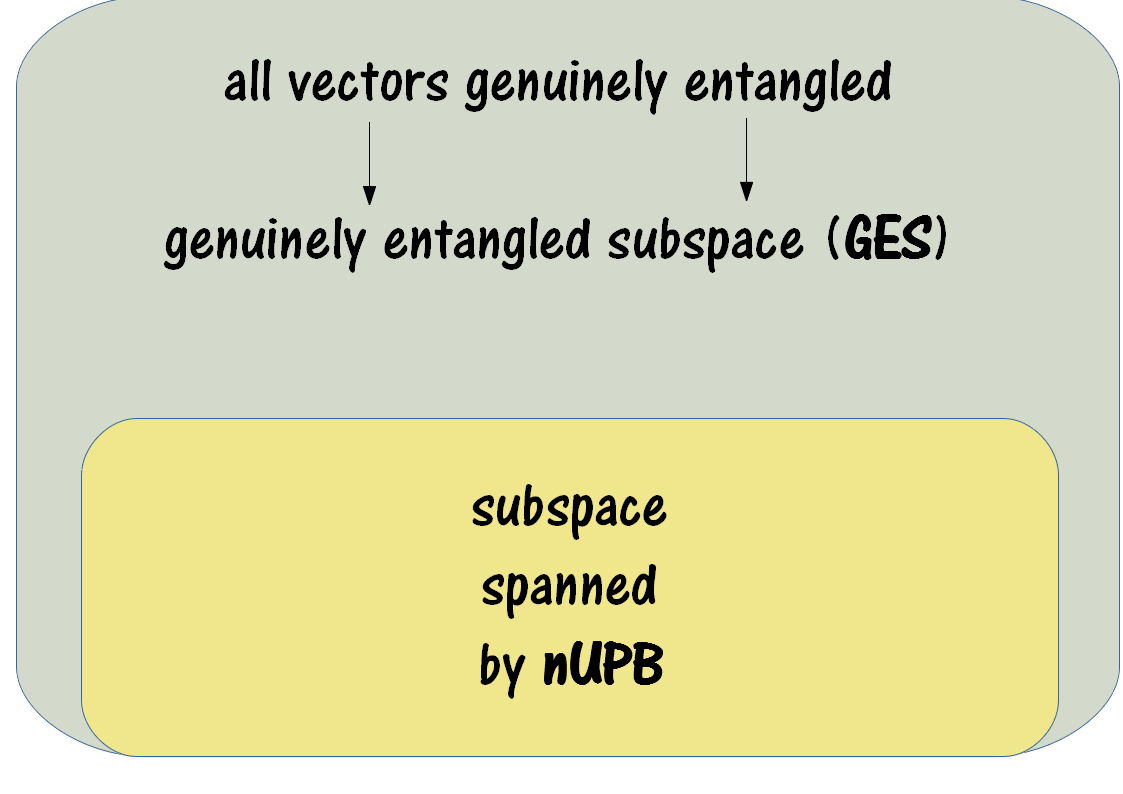}
\caption{
As discussed in the introduction, any $N$--partite Hilbert space with local dimensions $d_i$, $\calH_{d_1,\dots,d_2}=\mathbbm{C}^{d_1}\otimes \cdots \otimes \mathbbm{C}^{d_N}$, can be decomposed into a direct sum of a subspace spanned by a non--orthogonal unextendible product basis (nUPB) and a genuinely entangled subspace (GES), that is a subspace in which all states are genuinely entangled. 
In Theorems \ref{smaller-dim-ges}-\ref{produktowe-lepsze} we consider three different constructions of product bases with the above property. Interestingly, in cases when there are qubit subsystems the non--orthogonality assumption cannot be abandoned (see Section \ref{smaller-gesy}).}\label{podprzestrzenie-rys}
\end{figure}

The paper is organized as follows.  Sec. \ref{preliminary} recalls the terminology relevant for the following parts. In Sec. \ref{giesy} we formally introduce and discuss the notion of a genuinely entangled subspace. Then, in Sec. \ref{smaller-gesy}, 
we show how one can construct such subspaces as ones that are orthogonal to the spans of non--orthogonal unextendible product bases. Further, in Sec. \ref{splatanie-gesy}, we show how  our approach allows for an effortless construction of  genuinely entangled mixed states for any number of parties and an arbitrary local dimension on each site. We also discuss the issue of detecting entanglement of such states with the use of entanglement witnesses. Finally, we conclude in Sec. \ref{konkluzje} with some open questions and an outlook on further research directions.

\section{Preliminaries} \label{preliminary}

We begin with an introduction of the relevant notation and  terminology.


\paragraph{Notation.} In what follows we will be concerned with finite dimensional $N$-partite product Hilbert spaces
\begin{equation}\label{multipartyH}
\mathcal{H}_{d_1,\ldots,d_N}=\mathbbm{C}^{d_1}\otimes\ldots\otimes \mathbbm{C}^{d_N},
\end{equation}
with $d_i$ standing for the dimension of the local Hilbert space corresponding to the system $A_i$; we also use the shorthand 
$A_1 A_2\dots A_N =: \bf{A}$ to denote all subsystems. Pure states are traditionally denoted as $\ket{\psi},\ket{\varphi},\cdots$, potentially bearing subscripts corresponding to respective subspaces, \np, $\ket{\psi}_{A_1A_2\dots}$. 
Column vectors are  simply wrtitten as $(a,b,\dots)$. That is we write $\ket{\cdot}=(\dots)$, omitting for clarity the transposition. We also use the standard notation for tensor products of basis vectors: $\ket{ij}:=\ket{i}\otimes \ket{j}$.

\paragraph{Entanglement.} An $N$--partite pure state $\ket{\psi}_{A_1\dots A_N}$ is said to be {\it fully product} if it can be written as 
\begin{equation}
\ket{\psi}_{A_1\cdots A_N}=\ket{\varphi}_{A_1}\otimes\cdots \otimes\ket{\xi}_{A_N}.
\end{equation}
Otherwise it is {\it entangled}. Among such states there is one distinguished class being our main interest in the present paper, namely genuinely multiparty entangled  ones. A multipartite pure state is called {\it genuinely multiparty entangled} (GME) if 
\begin{equation}
\ket{\psi}_{A_1\cdots A_N}\ne\ket{\varphi}_{S}\otimes \ket{\phi}_{\bar{S}}
\end{equation}
for any bipartite cut (bipartition) $S | \bar{S}$, where $S$ is a subset of $\textbf{A}$ and $\bar{S}:=\textbf{A}\setminus S$ denotes the rest of them. Probably the most well-known example of a GME state is the $N$-qubit Greenberger-Horne-Zeilinger state \cite{GHZ} defined as
\begin{equation}\label{ghz}
\ket{\mathrm{GHZ}_N}=\frac{1}{\sqrt{2}}\left(\ket{0}^{\otimes N}+\ket{1}^{\otimes N}\right).
\end{equation}

On the other hand, if a state does admit the form  $\ket{\psi}_{A_1\cdots A_N} = \ket{\varphi}_{S}\otimes \ket{\phi}_{\bar{S}}$ it is called {\it biproduct}. Fully product states are thus a subclass of the biproduct ones.
 
Moving to the mixed states domain, one says that a state $\rho_{\textbf{A}}$ is {\it biseparable} if it can be written as
\beqn
\rho_{\textbf{A}}=\sum_{S|\bar{S}} p_{S|\bar{S}} \sum_i   q_{S|\bar{S}}^i \varrho^i_{S} \otimes \sigma^i_{\bar{S}},
\eeqn
with $\varrho^i_{S} $ and $\sigma^i_{\bar{S}}$ acting on, respectively, $\calH_{S}$ and $\calH_{\bar{S}}$, Hilbert spaces corresponding to a bipartite cut $S|\bar{S}$.
If a state does not admit such decomposition it is said to be {\it genuinely multiparty entangled} (GME), just as in the case of pure states.

\paragraph{Completely entangled subsaces and unextendible product bases.} 
We start off with a formal definition of completely entangled subspaces. 
\begin{defi}
 A subspace $\calC \subset \calH_{d_1,\dots , d_N}$ is called  a completely entangled subspace (CES) if all $\ket{\psi} \in \calC$ are entangled.
\end{defi}
It is worth stressing that the definition does not specify the type of entanglement of the states. It simply requires them not to be fully product.  

Completely entangled subspaces have been a subject of intensive studies in the literature \cite{ces-bhat,ces-partha,optimal,Sengupta2014,Brannan2017}. In particular, in Refs. \cite{ces-bhat,ces-partha} 
the maximal size of a CES in $\mathcal{H}_{d_1,\ldots, d_N}$
has been shown to be given by 
\begin{equation}
\prod_{i=1}^N d_i-\sum_{i=1}^N d_i+N-1
\end{equation}
and the corresponding examples of the subspaces were constructed. 

The notion of a completely entangled subspace is closely related to the notion of an unextendible product basis \citep{upb-bennett}. The definition of the latter is the following.
\begin{defi}\label{def:upb}
Let there be given a set of fully product vectors  
\begin{equation}\label{UPB}
U=\left\{\ket{\psi_i}\equiv \ket{\varphi_i}_{A_1}\otimes\ldots\otimes \ket{\xi_i}_{A_N}\right\}_{i=1}^u,
\end{equation}
$\ket{\psi_i}\in \calH_{d_1,\dots , d_N}$, with the property that it spans a proper subspace of $\calH_{d_1,\dots , d_N}$, i.e., $u<\dim\mathcal{H}_{d_1,\ldots,d_N}$, and no fully product vector exists in the complement of its span. Then,
if $\ket{\psi_i}$'s are mutually orthogonal, $U$ is called an \textbf{orthogonal unextendible product basis (oUPB)}. On the other hand, if the members of $U$ do not share this property, $U$ is called a \textbf{non--orthogonal unextendible product basis (nUPB)}. 
\end{defi}
In cases when the orthogonality property in a product basis is not particularized, we use a general term unextendible product basis (UPB)  encompassing both possibilities, {\it i.e.}, the basis can be either orthogonal or non--orthogonal.
Moreover, when $N=2$ (or the parties are simply split into two groups) we speak of bipartite UPBs, otherwise -- multipartite ones.

For an illustration of an oUPB consider the three-qubit 
Hilbert space $\mathcal{H}_{2,2,2}=(\mathbbm{C}^2)^{\otimes 3}$
and the following set of fully product vectors
\begin{equation}\label{upb-3kubity}
U=\{\ket{000},\ket{1\overline{e}e},\ket{e1\overline{e}},\ket{\overline{e}e1}\},
\end{equation}
where $\{\ket{0},\ket{1}\}$ and $\{\ket{e},\ket{\overline{e}}\}$
are two different orthonormal bases in $\mathbbm{C}^2$. It is not difficult to see that for any $\ket{e}\neq \ket{0},\ket{1}$ this set indeed forms a four-element oUPB, \tzn, there is no fully product vector orthogonal to every member of $U$, which itself is composed of orthogonal vectors.  Notice, however, that there does exist a biproduct vector orthogonal to all of the members of this basis. For example, a vector of this kind is given by $\ket{1}\otimes \ket{\xi}$, where $\ket{\xi}$ is orthogonal to the span of $ \{ \ket{\bar{e}e},\ket{1\bar{e}},\ket{e1}\}$. It is useful to keep in mind this observation for further purposes.

Let us now move to the case of non--orthogonal basis vectors and consider
the following set of, this time, bipartite vectors from $\mathcal{H}_{d,d}=\mathbbm{C}^d\otimes \mathbbm{C}^d$
given by 
\begin{equation} \label{pre-nUPB}
U'=\{\ket{e}\otimes \ket{e}\,|\,\ket{e}\in\mathbbm{C}^d\}. 
\end{equation}
In other words, the set consists of all product symmetric vectors from $\mathcal{H}_{d,d}$ and 
it is not difficult to see that $\mathrm{span}\:U'$ is simply the symmetric subspace of
$\mathcal{H}_{d,d}$. The subspace orthogonal to $\mathrm{span}\:U'$, being
the antisymmetric subspace of $\mathcal{H}_{d,d}$, 
is completely entangled as, quite trivially, it does not contain any 
product vector. The set $U'$ thus has the property of unextendibility required by Definition \ref{def:upb}.  However, it is not a basis yet as it has more elements (in fact, it has an infinite number of them) than the dimension of the subspace spanned by them. 
Selecting ${d+1 \choose 2}$, the dimension of the bipartite symmetric subspace, linearly independent symmetric vectors, we can turn $U'$ into a basis, which here is non--orthogonal and unextendible, {\it i.e.}, it is an nUPB. 
Using the Gram--Schmidt procedure one can make the chosen vectors orthogonal. Importantly, however, some of them will necessarily be entangled.  As an illustration to the construction considered in this paragraph, consider the bipartite qubit case ($d=2$). Then, $\dim \spann \;U'={3 \choose 2} =3$ and we choose this number of product vectors to construct the related nUPB. For example, one can take $\{\ket{0}\ket{0}$, $\ket{1}\ket{1}$, $\ket{+}\ket{+}\}$, where $\ket{+}=1/\sqrt{2}(\ket{0}+\ket{1})$.
Now, the orthogonalization produces another basis for $\spann\; U'$: $\{\ket{00},\ket{11},\ket{01}+\ket{10}\}$, which, however, is not product anymore. This example is interesting in that there is no oUPB at all in $\cee{2}\otimes \cee{2}$ (even more generally, in $\cee{2}\otimes \cee{d}$; we comment on the consequences of this for our results in Section \ref{smaller-gesy}).

 Non--orthogonal UPBs have been considered in Refs. \cite{upb-pittenger,ces-partha,ces-bhat,LeinaasUPB,Skowronek} and are the main scope of the present paper.


The crucial observation linking the notions of completely entangled subspaces and unextendible product bases is that the orthogonal complement of a subspace spanned by a UPB, whether its members are mutually orthogonal or not, is a CES. Notice, however, that the implication in the opposite direction is not true in general \cite{WalgateScott,Skowronek}. That is, the orthocomplement of a CES does not necessarily admit a UPB, neither orthogonal nor non--orthogonal. Actually, an even stronger result holds: the orthocomplement of a CES can be a CES itself \cite{Skowronek,WalgateScott}. 
\section{Genuinely entangled subspaces}\label{giesy}

It is obvious that while fully product states are absent in a CES, there still might be present other biproduct states. 
This basic observation motivates an introduction of the notion of subspaces void of any biproduct states or, in other words, subspaces in which entanglement is solely of the genuinely multiparty nature \cite{schmidt-rank}.
We propose the name
{\it genuinely entangled subspaces} for them. Their formal definition  is as follows.

\begin{defi} 
A subspace $\calG \subset \calH_{d_1,\dots , d_N}$ is called a genuinely entangled subspace (GES)  of $\calH_{d_1,\dots , d_N}$ if all $\ket{\psi} \in \calG$ are genuinely multiparty entangled.
\end{defi}
Clearly, every GES is also a CES. However, the opposite implication is generally 
not true and therefore the set of all genuinely entangled subspaces is a proper subset of
the set all completely entangled ones.
Drawing from the terminology used in the previous section, one could say that a GES is such a subspace that is entangled across any of the cuts. Equivalently, one could also say that it is a subspace in which all states are of Schmidt rank at least two across any of the bipartite cuts (cf. \cite{schmidt-rank}).

A well known example of such a subspace is the already mentioned antisymmetric space in the Hilbert space of $N$ qudits. 
These subspaces, however, as Observation \ref{wymiary-ges} below shows, are quite small in the sense that their dimensionality, which is ${d \choose N}$, is relatively far from the maximal dimension available for a GES. 
In an extreme case of $N >d$, the antisymmetric subspace is empty, while there are always nontrivial  GESs (with dimension larger than one) in any dimension.

A fundamental question arises about how the additional constraint about the genuine multiparty entanglement of the states in a GES determines its maximal possible dimension, denote it $D_{\mathrm{max}}^{GES}$.
The following simple observation gives a complete answer \cite{schmidt-rank}.

\begin{obserwacja}\label{wymiary-ges}
Given $\calH_{d_1,\dots,d_N}$, with $2 \le d_i \le d_{i+1}$, the maximal achievable  dimension of a GES is 
\beq\label{dim-ges-max}
D_{\mathrm{max}}^{\mathrm{GES}}=\prod_{i=1}^{N} d_i-(d_1+d_2\cdot d_3\cdot \ldots \cdot d_N)+1.
\eeq
\end{obserwacja}
In fact, a randomly chosen subspace
of dimension $D_{\mathrm{max}}^{\mathrm{GES}}$ in $\mathcal{H}_{d_1,\ldots, d_N}$ will typically be genuinely entangled \cite{schmidt-rank}. From a perspective more relevant for the current approach, as discussed in the introduction, also a set of $(d_1+d_2\cdot d_3\cdot \ldots \cdot d_N)-1$ fully product vectors will have in the orthocomplement of their span a GES of the above dimension. The argument holds for other dimensions of GESs too.

In this paper, we mainly concentrate, for simplicity, on the case of equal dimensions, \tzn, 
$d_1=\ldots =d_N=d$, in which case $\calH_{d_1,\dots,d_N}$ is denoted simply as
$\calH_{N,d}$. Then,
\beq \label{max-dim-d}
D_{\mathrm{max}}^{GES}=(d^{N-1}-1)(d-1).
\eeq

It is worth analyzing two limits of the above dimension: the increasing local dimensions and the increasing number of parties. It holds that for large $d$  the dimension of a maximal GES tends to the dimension of the full space, while for large number of parties $N$ the fraction  $D_{\mathrm{max}}^{GES}/ \dim \calH_{d^N}=
(d^{N-1})(d-1) / d^N$ goes to $1-\frac{1}{d}$.
\section{Genuinely entangled subspaces from unextendible product bases}\label{smaller-gesy}

We now move to the main body of the present work, where we consider the problem of a general  construction of genuinely entangled subspaces from unextendible product bases. 

Let us begin with a simple but crucial general observation.
\begin{uwaga}\label{remark}
A multipartite UPB has a GES in the orthocomplement of its span if and only if it is a bipartite UPB across any of the possible cuts in the parties.
\end{uwaga}
 This means that although we focus on the general $N$--party case, our considerations in fact reduce to repeated analyses of the two--party instances of the problem and we can make use of the tools developed for this case. 

Remark \ref{remark} implies, in particular, that in cases when at least one of the parties holds a qubit system, no oUPB can give rise to a GES. This stems from a well--known fact that there do not exist bipartite oUPBs in $2 \otimes d$ systems \cite{upb-bennett}. The same cannot be said if all $d_i \ge 3$ and, in fact, there are constructions available in these setups \cite{cmp-upb,NisetCerf}. Still, to our knowledge, no already-known oUPB  defines a GES in its complement. 
Furthermore, as we argued before, there exist genuinely entangled subspaces of arbitrary dimensions, and they can be obtained from nUPBs by a random draw of multiparty fully product states \cite{cmp-upb,Skowronek}. On the other hand, it is known that oUPBs cannot exist with any, {\it a priori} accessible, cardinality. This limits the possible range of applicability of any potential approach to the construction of GESs based on oUPBs. It might also well be the case that such an approach  is excluded for fundamental reasons. We do not have, however, enough evidence to support any of the  cases and we leave this problem open here. 

With the goal being a general construction working for any number of parties $N$ and local dimensions $d_i$, including $d_i=2$, we thus look into the case of nUPBs in the search of UPBs giving rise to GESs. We obtain both small dimensional GESs and large ones as well.

We will need an observation concerning spanning properties of tuples of local vectors stemming from sets of product vectors. The following holds.
\begin{lem}\label{glowny-lemat}  \cite{upb-bennett} (see also \cite{upb-pittenger,Skowronek})
Let there be given a set of product vectors $B=\{ \ket{\varphi_x}\otimes \ket{\phi_x} \}_{x}$\; from $\cee{m}\otimes\cee{n}$ with cardinality $|B| \ge m+n-1$. 
If any $m$--tuple of vectors $\ket{\varphi_x}$ spans $\cee{m}$ and any $n$--tuple of $\ket{\phi_x}$'s spans $\cee{n}$, then there is no product vector in the orthocomplement of $\spann B$.
\end{lem}
It is instructive to realize why this is true.
If we could partition $B$ as $B=B_1\cup B_2$, so that  the local rank of $B_1$ as seen by the first party was strictly smaller than $m$, and similarly for $B_2$ -- its local rank as seen by the second party was strictly smaller than $n$, then it would be possible to find a product vector $\ket{f}\otimes\ket{g}$
in $(\spann B)^{\perp}$, the orthocomplement of $\spann B$. 
We could then just take $\ket{f}$ orthogonal to the span of $\ket{\varphi_x}$'s appearing in $B_1$ and $\ket{g}$ orthogonal to the span of $\ket{\phi_x}$'s appearing in $B_2$. 
With the properties of $\ket{\varphi_x}$'s and $\ket{\phi_x}$'s as given by the lemma it is clearly not possible to find such a partition: for at least one set in any partition the local rank will attain the dimension of the local space. 


 We will refer to the properties of local vectors specified by Lemma \ref{glowny-lemat} shortly as to the {\it spanning}. Since it is a very important notion for the remainder of the paper  we single out its formal definition.
\begin{defi}\label{def-span}
 Given a set of vectors $\ket{x_i}_S\otimes \ket{y_i}_{\bar{S}}\in\cee{m}\otimes\cee{n}$, it is said that the spanning on $S$ holds for this set if any $m$--tuple of vectors $\ket{x_i}$ spans $\cee{m}$. Similarly, the spanning on $\bar{S}$ holds if any $n$--tuple of $\ket{y_i}$'s spans $\cee{n}$. In other terms, $\ket{x_i}$'s and $\ket{y_i}$'s possess the spanning property.
\end{defi}
%


 Sets of product vectors with the spanning property on both subsystems can be easily constructed with the aid of vectors being rows of Vandermonde matrices, that is vectors of the form
 \beqn
 \ket{v_p(a)}= \sum_{i=0}^{p-1}a^i\ket{i}=(1,a,a^2,a^3,\dots,a^{p-1})\in \cee{p} .
 \eeqn
  We will call them {\it Vandermonde vectors}. Such vectors, share, exactly as we need, the property that any $r$--tuple ($r \le p$) of them with $r$ different values of $a$ spans an $r$--dimensional subspace of $\cee{p}$. We can thus take
\beqn\ket{\varphi_i}=\ket{v_m(\lambda_i)}, \quad \ket{\phi_i}=\ket{v_n(\lambda_i)},
\eeqn
 with some arbitrary  $\lambda_i$'s such that $\lambda_i \ne \lambda_j$ for $i \ne j$, to construct a set of vectors 
 \beqn
 \{ \ket{\varphi_i}\otimes \ket{\phi_i} \}_{i=1}^s, \quad s\ge m+n-1,
 \eeqn
  for which the spanning holds on both subsystems. By Lemma \ref{glowny-lemat} we then conclude that the subspace orthogonal to the span of these vectors is void of product vectors, in other words, it is completely entangled.

In fact, this type of reasoning transfers without basically any changes to the multiparty case. More precisely, one constructs a set 
\beqn
\{\ket{\psi^{(1)}_i}\otimes \cdots \otimes \ket{\psi^{(N)}_i}\}_{i=1}^{s},\quad \ket{\psi_i^{(j)}}=\ket{v_{d_j}(\lambda_i)},\nonumber\\
\hspace{-3.5cm}\quad s \ge \sum_{j=1}^N d_j+N-1,
\eeqn
 with 
 $\lambda_i\ne \lambda_j$ for $i \ne j$,
 for which $\ket{\psi_i^{(j)}}$'s have the spanning property for each $j$. An argument virtually the same as the one given as a justification of Lemma \ref{glowny-lemat} can be applied here and one concludes that the subspace orthogonal to the span of the vectors given above is completely entangled, {\it i.e.}, it is void of {\it fully} product vectors \cite{ces-partha}.

The described method as it stands cannot be, however, applied  to a construction of genuinely entangled subspaces for the following reason. Take the set of vectors as above with $s \ge \prod_{j=1}^N d_j-D_{\mathrm{max}}^{\mathrm{GES}}$ [the lower bound is the minimal number necessary; see Eq. (\ref{dim-ges-max}) for the value of $D_{\mathrm{max}}^{\mathrm{GES}}$]. Recall that a GES must be a CES when considered in any bipartite cut.  Consider any such cut, {\it e.g.}, $A_1 A_2 | A_3 \dots A_N$. Locally on subsystem $A_1A_2$ the vectors are given by $\ket{\psi_i^{(1)}}\otimes \ket{\psi_i^{(2)}}=\sum_{k=0}^{d_1-1}\sum_{l=0}^{d_2-1} \lambda_i^k \lambda_i^l \ket{k}\ket{l}$. Clearly, since we have repeating powers of $\lambda_i$ in the coordinates of these vectors they do not
 span the whole space on $A_1 A_2$. It then easily follows that we can find a vector $\ket{f}_{A_1A_2}$ orthogonal to these vectors, which, in turn, implies that there is a product vector $\ket{f}_{A_1A_2}\otimes \ket{g}_{A_3\dots A_N}$, with $\ket{g}$ being arbitrary, orthogonal to any of the spanning vectors. As such, the CES under consideration is not a GES. 

A seemingly straightforward way out would be to use different sets of numbers $\{\lambda_{ij}\}$ for each subsystem $j$, instead of using the same set $\{\lambda_i\}$ for every party. Nevertheless, these sets would have to be very carefully chosen to guarantee the spanning property locally for any bipartition (note that here we are talking about a particular approach based on Lemma \ref{glowny-lemat}, which only is a sufficient condition for a product basis to be unextendible). Without any hint about how to do it this seems a formidable task in the general case of arbitrary local dimensions and number of parties.
It should be noted though that random sets of $\lambda_{ij}$'s in principle would do the job,  but then the construction would not be much different than just taking a random GES. 
 Such subspaces are not in the range of our interest since we are concerned with GESs with well defined structures as they are subsequently utilized in constructions of GME states (see Section \ref{splatanie-gesy}).

To circumvent the difficulties exposed above we  put forward a different approach, in which basis vectors have, by construction, the spanning property  (see Definition \ref{def-span}) locally  for any bipartite cut. By Lemma \ref{glowny-lemat}, this implies that the orthocomplement to the span of such vectors is a GES.


Let us now give an  overview of this method.
As indicated earlier, we concentrate on the case of equal local dimensions, but the methodology remains the same for other cases.


  We consider continuous sets of fully product vectors %
\beq\label{upb-alfa}
\calB=\{\ket{\Psi(\alpha)}\equiv\bigotimes_{k=1}^{N}\ket{\psi_k(\alpha)}_{A_k}\; |  \alpha\in\ce\},
\eeq
with the local states $\ket{\psi_k(\alpha)}\in \mathbb{C}^{d}$ assumed to have coordinates being either monomials or polynomials of $\alpha$. They   are chosen in such a way that the coordinates of the vectors
\beqn\label{podzbiory}
\bigotimes_{k\in I}\ket{\psi_k(\alpha)}_{A_k},\quad I\subset \{1,2,\dots, N\},
\eeqn
are linearly independent  functions of $\alpha$ for any $I$, which ensures that locally, for any partition, the vectors span corresponding whole spaces on subsystems. As we have already realized, this precludes using Vandermonde vectors $v_d(\alpha)$ directly as $\ket{\psi_k(\alpha)}$'s: tensor products of Vandermonde vectors have repeating  monomials of $\alpha$ in the coordinates and thus such constructed vectors do not span whole spaces of the subsystems.  In principle, the linear independence  is only a necessary condition if one wants to construct a UPB. Here, however, it  also  turns out sufficient. The argument goes as follows. Let $u$ be the dimension of the subspace spanned by the vectors from $\calB$, \tzn,
\beq \label{dim-u}
u=\dim \spann\; \calB.
\eeq
Since $\calB$ is a continuous set we can choose $u$ values of $\alpha$ so that the vectors from the set
\beqn \label{baza}
\bar{\calB}=\{ \ket{\Psi_i}\equiv   \ket{\Psi(\alpha_i)} \}_{i=1}^{u},
\eeqn
$\alpha_i\in \ce$, span the same subspace as those from $\calB$, and locally have the spanning property  for any bipartite cut. Due to Lemma \ref{glowny-lemat},  there is no biproduct vector in the orthocomplement of $\spann\;\bar{\calB}$, meaning that $\bar{\calB}$ is a UPB giving rise to a GES. The details of the derivation are given in Appendix \ref{App-A}.

The procedure discussed above makes a direct correspondence between the sets $\calB$ and $\bar{\calB}$. For this reason, we will identify UPBs  with the sets $\calB$ from (\ref{upb-alfa}), as the latter provide a compact description of the corresponding UPBs.
 
%
%

Common to our constructions is the form of $\ket{\psi_k(\alpha)}$'s for $k=2,\dots,N$, which is
\beqn\label{common}
\ket{\psi_k(\alpha)}=(1,\alpha ^{d^{N-k}},\alpha ^{2d^{N-k}},\dots,\alpha ^{(d-1)d^{N-k}}).
\eeqn
The following then holds:
\beqn\label{produktowy-rozpis}
\bigotimes_{k=2}^{N}\ket{\psi_k(\alpha)}_{A_k}&=& \bigotimes_{k=2}^{N} (1,\alpha ^{d^{N-k}},\alpha ^{2d^{N-k}},\dots,\alpha ^{(d-1)d^{N-k}})_{A_k}\non &=&  (1,\alpha,\alpha^2,\alpha^3,\cdots,\alpha^{d^{N-1}-1})_{A_2A_3\cdots A_N}\non
&=& \ket{v_{d^{N-1}}(\alpha)}_{A_2A_3\cdots A_N}.
\eeqn
That is, instead of using Vandermonde vectors on each party, we use them on the $(N-1)$-partite subsystem $A_2A_3\dots A_N = \textbf{A}\setminus A_1 $ of all the parties.
Since the entries of (\ref{produktowy-rozpis}) are linearly independent monomials,  with such a choice,  we guarantee  that also the coordinates of all the vectors of the type (\ref{podzbiory}) with $I\subset \{2,3,\dots,N\}$ are  linearly independent monomials of $\alpha$ (such result is true whenever linearly independent functions are involved, see the following discussion and Appendix \ref{App-B}). 
We then consider different choices for
$\ket{\psi_1(\alpha)}$, such that the linear independence of coordinates considered as functions of $\alpha$  of the proper vectors also holds on every proper subset of all the parties, including the party $A_1$. This ensures, as discussed above (see also Appendix \ref{App-A}), that the spanning in the derived basis (\ref{baza}) holds for any bipartition and we can make use of Lemma \ref{glowny-lemat} to infer that a given UPB leads to a GES. Importantly, to show linear independence on subsystems containing $A_1$ we do not need to consider all such subsystems -- it is sufficient to consider only $(N-1)$--partite ones and the result for the ones with a smaller number of parties then follows; this quite obvious result is given for the ease of reference in Appendix \ref{App-B}. Whilst the proof of Theorem \ref{smaller-dim-ges} does not refer to this observation, proofs of Theorems \ref{N-kuditow}-\ref{produktowe-lepsze} heavily exploit it to reduce the effort in computation.

While we care about linear independence of the coordinates, at the same time we require the condition $\dim \spann \calB < d^N$ to hold, that is the resulting GES to be nonempty.


 To compute the dimension of the latter we first find $u$ [Eq. (\ref{dim-u})] by counting linearly independent functions of $\alpha$ in the  coordinates of  $\ket{\Psi(\alpha)}$'s, and then substract it  from $d^N$, the dimension of the full Hilbert space $\calH_{\textbf{A}}$.  It is thus an important task to make the number $u$ the smallest possible, so that the arising GES is large.

\subsection{Monomial coordinates of vectors}
We first look into the case of monomial coordinates of the vectors in an nUPB.

We begin with a simple, we might even call it brute--force, construction
of  GESs of small dimensionality constant in $N$.
\begin{theorem}\label{smaller-dim-ges}
Let $V_1$ be the following set of product vectors from 
$(\mathbbm{C}^d)^{\otimes N}$:
\beq\label{smaller-spanujace}
V_1=\left\{\ket{\psi_1^{(1)}(\alpha)}_{A_1}\otimes\ket{\psi_2(\alpha)}_{A_2}\otimes\ldots\otimes\ket{\psi_N(\alpha)}_{A_N}\,|\,\alpha\in\mathbbm{C}\right\},
\eeq
where $\ket{\psi_k(\alpha)}$, $k=2,\ldots,N$, are given by 
Eq. (\ref{common}), whereas $\ket{\psi_1^{(1)}(\alpha)}$ is defined through 
\beq
\ket{\psi_1^{(1)}(\alpha)}=\big(1,\alpha^{\widetilde{d}},\alpha^{2\widetilde{d}}\dots,\alpha^{(d-1)\widetilde{d}}\big)
\eeq
with
\beq
\widetilde{d}:=\sum_{k=2}^{N-1}(d-1)d^{N-k}+1=d^{N-1}-d+1.
\eeq
Then, the subspace orthogonal to $\mathrm{span}(V_1)$ is a GES
of dimension $(d-1)^2$.
\end{theorem}
\begin{proof}
We first prove that the subspace  is indeed genuinely multiparty entangled. With this aim, it is enough, as we argued above, to show linear independence of the  coordinates (as functions of $\alpha$) of the vectors for all bipartite cuts of the parties. 
Consider a bipartition  $S|\bar{S}$,  assuming w.l.o.g. that $\bar{S} \subseteq \textbf{A}\setminus A_1$. Write the vectors from $V_1$  with respect to such bipartition as:
\beqn \label{bipartycja}
\ket{\phi(\alpha)}_S \otimes \ket{\varphi(\alpha)}_{\bar{S}},
\eeqn
where $\ket{\phi(\alpha)}_S = \bigotimes_{A_i\in S} \ket{\psi_i(\alpha)}_{A_i}$ and $\ket{\varphi(\alpha)}_{\bar{S}} = \bigotimes_{A_i\in \bar{S}} \ket{\psi_i(\alpha)}_{A_i}$.
 As already observed (see also Appendix \ref{App-B}),  the  coordinates of  $\ket{\varphi(\alpha)}_{\bar{S}}$ (being monomials in $\alpha$) are linearly independent.
The same will now be proved for subsystem $S$. When $\bar{S}=\textbf{A}\setminus A_1$ the subsystem $S$ simply is $A_1$ and we trivially have the desired result. Consider now the case $\bar{S} \subset \textbf{A}\setminus A_1$. By construction, $\widetilde{d}$ is greater than  powers of $\alpha$ in the coordinates of any vector on  a subsystem of $A_2\dots A_{N}$, with the largest of these powers being $\widetilde{d}-1$ [corresponding to the $(N-2)$--partite subsystem $A_2\dots A_{N-1}$]. Clearly, due to this reason, after multiplying any such vector 
by $(1,\alpha^{\widetilde{d}},\dots,\alpha ^{(d-1)\widetilde{d}})$ on $A_1$ to obtain $\ket{\phi(\alpha)}_S$, there will be no repeating, 
{\it i.e.,} linearly dependent, monomials of $\alpha$ in the entries of the latter.  This concludes this part of the proof.

As to the dimension of the GES,  this is, as announced earlier, just a simple counting of linearly independent monomials  of $\alpha$ in the entries of the  vectors from the set $V_1$. Writing down explicitly the monomials being the coordinates of these vectors in the order of increasing powers may be useful with this aim:
\beqn
&& 1,\alpha,\dots,\alpha^{\widetilde{d}},\dots, \alpha^{d^{N-1}-1}, \dots, \alpha^{2\widetilde{d}},\dots, \alpha ^{\widetilde{d}+d^{N-1}-1},\dots, \non
&&\hspace{+1.5cm} \alpha^{3\widetilde{d}},\dots, \alpha^{2\widetilde{d}+d^{N-1}-1},\dots,\alpha^{d^N-(d-1)^2-1}.
\eeqn
Clearly, all the powers of $\alpha$ up to the value $(d-1)\widetilde{d}+(d-1)\sum_{k=2}^N d^{N-k}=d^N-(d-1)^2-1$ appear. In turn,  $\dim \spann \; V_1=d^N-(d-1)^2-1+1=d^N-(d-1)^2$ and the dimension of the GES is $d^N-\spann \; V_1$. This concludes the proof.

\end{proof}

The construction above is, as it could have been predicted, quite far from being optimal regarding the dimensionality it achieves and a significant improvement of the performance can be achieved.
With this respect, the next one not only recovers the dependence on the number of parties but gives, except the special case of $d = 2$, the ,,correct'' order, $\sim d^N$, of the leading term  in the dimension as well. It is given by the following theorem.
\begin{theorem}\label{N-kuditow}
Let $V_2$ be the following set of product vectors from 
$\calH_{N,d}$:
\beqn\label{N-kuditow-wektory}
V_2=\left\{\ket{\psi_1^{(2)}(\alpha)}_{A_1}\otimes\ket{\psi_2(\alpha)}_{A_2}\otimes\ldots\otimes\ket{\psi_N(\alpha)}_{A_N}\,|\,\alpha\in\mathbbm{C}\right\}, \non
\eeqn
where $\ket{\psi_k(\alpha)}$, $k=2,\ldots,N$, are defined in
Eq. (\ref{common}), while $\ket{\psi_1^{(2)}(\alpha)}$ is of the form  
\beq
\ket{\psi_1^{(2)}(\alpha)}=\big(1,\alpha^{p_1},\alpha^{p_2},\dots, \alpha ^{p_{d-1}}\big)
\eeq
with
\beq
p_i:= \sum_{k=2}^N i d^{N-k}, 
\eeq
$i=1,2,\dots, d-1$. Then, the subspace orthogonal to $\mathrm{span}(V_2)$ is a GES
of dimension  $d^N-(2d^{N-1}-1)$.
\end{theorem}
Clearly, $p_i=i p_1$, nevertheless, in view of upcoming Theorem \ref{produktowe-lepsze}, we prefer to keep the denotations in the lemma as stated.
\begin{proof}
First, we prove that the subspace is genuinely entangled. Again, we consider bipartitions $S|\bar{S}$ with $\bar{S} \subset \textbf{A}\setminus A_1$ (the case $\bar{S}=\textbf{A}\setminus A_1$ is, just as before, trivial) and  prove linear independence of the coordinates (as functions of $\alpha$) of the resulting local vectors on $S$, since the vectors on $\bar{S}$ are the same as previously.
 This time, however, we exploit the observation that with this aim it is only enough to consider the cases with $|S|=N-1$ as the result for all subsystems of such $S$ then follows (see Appendix \ref{App-B}). Stating this differently, we consider all bipartitions such that $S=\textbf{A}\setminus A_j$ with $j=2,3,\dots,N$.

We then define
\beq \label{without-two}
\ket{\widetilde{f}_j(\alpha)}_{\textbf{A}\setminus A_1 A_j}:= \bigotimes_{\substack { k=2 \\ k\ne j}}^{N} (1,\alpha ^{d^{N-k}},\alpha ^{2d^{N-k}},\dots,\alpha ^{(d-1)d^{N-k}})_{A_k}
\eeq
%
and verify linear independence of the coordinates of the following vectors
\beqn \label{produkt-ogolny}
\ket{\psi_1^{(2)}(\alpha)}_{A_1}\otimes \ket{\widetilde{f}_j(\alpha)}_{\textbf{A}\setminus A_1 A_j}, \quad j=2,\dots,N,
\eeqn 
%
by showing that all  monomials arising in (\ref{produkt-ogolny}) are different.

Each monomial in the entries of the vector on $\textbf{A}\setminus A_j$ can be represented as $\alpha^{\frak{g}+sp_1}$, $s=0,1,\dots,d-1$, with
\beqn
\frak{g}=\sum_{\substack { k=2 \\ k\ne j}}^N i_{N-k} d^{N-k}, \quad i_{N-k}=0,1,\dots, d-1.
\eeqn

If all the monomials were not unique, there would be a pair $(\frak{g},s)$ with a another related solution $(\frak{g}',s')$, \tzn,
\beq\label{warunek-rownosci}
\frak{g}+sp_1=\frak{g}'+s'p_1,
\eeq
where
\beq
\frak{g}'=\sum_{\substack { k=2 \\ k\ne j}}^N i_{N-k}' d^{N-k}, 
\eeq
with $ i_{N-k}'\in \{0,1,\dots, d-1\}$ and $s'\in \{0,1,\dots, d-1\}$. We can assume $s' \ge s$. The claim is that there is no nontrivial, that is different than the original unprimed one, solution to this.

The condition (\ref{warunek-rownosci}) translates into the statment that there  exist $\{i_{N-m}\}, \{i'_{N-k}\}$ with $m\in \{ 0,1,\dots, d-1\}$ such that
\beqn
\label{warunek-sumy}
\sum_{\substack { k=2 \\ k\ne j}}^N (i_{N-k}-i'_{N-k}) d^{N-k}       =\sum_{k=2}^N m d^{N-k}.
\eeqn
The form of the numbers involved suggests conducting the remaining analysis using the  representation of numbers in the base-$d$. 
Rewriting (\ref{warunek-sumy}) using such representation we have
\beqn
\begin{array}{cccccccc}
&(&i'_{N-2}&\dots  &0  &\dots& i'_1& i'_0)_d\\ +&(&m&\dots &m &\dots & m &m)_d \\ =&(&i_{N-2}&\dots & 0 &\dots &i_1& i_0)_d
\end{array}
\eeqn
%
with $0$'s on the $(j-1)$--th positions corresponding to terms $d^{n-j}$.
Regardless of the base, while adding two numbers the carry from the $k$--th position (counting from the right) at the ($k+1)$--th position is always $0$ or $1$. This implies that $m$ could only be equal to $0$ or $d-1$. Clearly the latter solution is impossible, while the former leads to the same solution. Thence, monomials in the entries are unique and in consequence linearly independent.

We now find the dimension of the GES.  Let us expand the part of the spanning vectors from $V_2$ on $\textbf{A}\setminus A_1$:
\beqn
&&\bigotimes_{k=2}^{N} (1,\alpha ^{d^{N-k}},\alpha ^{2d^{N-k}},\dots,\alpha ^{(d-1)d^{N-k}})_{A_k} =\non
&&\hspace{+0.1cm}(1,\alpha,\dots,\alpha^{p_1-1},\alpha^{p_1},\dots, \alpha^{p_2-1}, \alpha^{p_2},\dots,\alpha^{p_{d-1}-1},\alpha^{p_{d-1}}).\non
\eeqn
All of these monomials are linearly independent. Additional linearly independent terms stem from the multiplications by $\alpha^{p_i}$ in $\ket{\psi_1^{(2)}}$. It is easy to realize that each multiplication introduces $p_1$ new terms. In turn, there is a total of $d^{N-1}+(d-1)p_1=2d^{N-1}-1$ linearly independent monomials in (\ref{N-kuditow-wektory}), which proves the claimed dimension of the GES.
\end{proof}

The construction achieves the maximal dimension within the approach taking monomial coordinates. It can be easily seen if one realizes that the construction in fact is as follows. We start with the vectors (\ref{produktowy-rozpis}). Then, the coordinates on $A_1$ populate available monomials starting with the lowest powers of $\alpha$ so that we keep spanning on any subset. The fact that the monomials on $A_1$ have the smallest possible degrees ensures that the number of different monomials in the coordinates of (\ref{N-kuditow-wektory}) is the smallest possible, in turn giving a GES of the largest dimension. 

Smaller dimensions can be obtained by varying the powers of monomials in the vector on $A_1$. We discuss this in Section \ref{generalki}.

Concluding this subsection, we note that in the qubit case both constructions coincide and single out only one GME state, which is of the form:
\beq\label{single-gme}
\ket{\Psi}=\frac{1}{\sqrt{2}}\left(
\ket{0}\ket{1}^{\otimes (N-1)}-\ket{1}\ket{0}^{\otimes (N-1)}\right).
\eeq
With a local operation $-\uroj \sigma_y $ on site $A_1$ the state can be transformed into the GHZ state. 

\subsection{Polynomial coordinates of vectors}\label{polynomial-coordinates}

Clearly, assuming the coordinates to be monomials is not by any means a general approach. In principle, allowing the entries to be polynomials in $\alpha$ might increase the dimension of GESs.

The construction providing evidence that this is indeed the case is the content of the upcoming theorem. It is in fact inspired by the one given in Theorem \ref{N-kuditow} and may be considered its generalization. We have the following.
\begin{theorem}\label{produktowe-lepsze}
Let $V_3$ be the following set of product vectors from 
$\calH_{N,d}$:
\beqn\label{produkty-nowe-lepsze}
V_3=\left\{\ket{\psi_1^{(3)}(\alpha)}_{A_1}\otimes\ket{\psi_2(\alpha)}_{A_2}\otimes\ldots\otimes\ket{\psi_N(\alpha)}_{A_N}\,|\,\alpha\in\mathbbm{C}\right\}, \non
\eeqn
where $\ket{\psi_k(\alpha)}$, $k=2,\ldots,N$, are given by
Eq. (\ref{common}), and
\beq
\ket{\psi_1^{(3)}(\alpha)}=\big (1,P_1(\alpha),P_2(\alpha),\dots, P_{d-1}(\alpha)\big)
\eeq
with
\beq
P_i(\alpha):= \sum_{k=2}^N\alpha^{ i d^{N-k}}, 
\eeq
$i=1,2,\dots, d-1$. Then, the subspace orthogonal to $\mathrm{span}(V_3)$ is a GES
of dimension  $d^{N-2} (d-1)^2=d^N-2d^{N-1}+d^{N-2} $.
\end{theorem}

Direct comparison shows that such constructed GESs have $d^{N-2}-1$ more elements than the ones from Theorem \ref{N-kuditow}.
\begin{proof}
We begin with the proof that the subspace is genuinely entangled. Following the same line of thought as in the proof of Theorem \ref{N-kuditow}, we only need to consider  bipartitions $S|\bar{S}$ with $S=\textbf{A}\setminus A_j$, $j=2,3,\dots,N$ as for for the remaining cases the result follows.

Denote 
\beq
\ket{f(\alpha)}_{\textbf{A}\setminus A_1 }:= \bigotimes_{ k=2 }^{N} (1,\alpha ^{d^{N-k}},\alpha ^{2d^{N-k}},\dots,\alpha ^{(d-1)d^{N-k}})_{A_k},
\eeq
and consider again $\ket{\widetilde{f}_j(\alpha)}_{\textbf{A}\setminus A_1 A_j}$
defined in Eq. (\ref{without-two}).
%
%
Similarly to the proof of Theorem \ref{N-kuditow}, we prove linear independence of the functions, here polynomials in $\alpha$, being the coordinates of the vectors 
\beqn
\ket{\psi_1^{(3)}(\alpha)}_{A_1}\otimes \ket{\widetilde{f}_j(\alpha)}_{\textbf{A}\setminus A_1 A_j}, \quad j=2,\dots,N,
\eeqn
which is sufficient to support the claim. 

We begin with some preparatory terminology.
Let
\beqn
g_m^j
&:=&\left(\alpha ^{(m-1)d^{N-j+1}},\dots,     \alpha ^{(m-1)d^{N-j+1}+d^{N-j}-1}     \right)
 \nonumber\\
&=&
\left(\alpha ^{md^{N-j+1}-d^{N-j+1}},\dots,     \alpha ^{md^{N-j+1}-(d-1)d^{N-j}-1}     \right),\non 
\eeqn
for $m=1,2,\dots,d^{j-2}$. We will refer to $g_m^j$'s as the ,,groups''. In terms of the groups, we can, with a little abuse of mathematical notation, write:
\beqn
 \ket{\widetilde{f}_j(\alpha)}_{\textbf{A}\setminus A_1 A_j}=(g_1^j,g_2^j,\dots, g_{d^{j-2}}^j)= \bigoplus_{m=1}^{d^{j-2}} g_m^j.
\eeqn
As one can see each $g_m^j$ has $d^{N-j}$ elements.

The ,,gaps'' are by definition the following $(d-1)d^{N-j}$ element sets
\beqn
\bar{g}_m^j
&=&\left(\alpha ^{(m-1)d^{N-j+1}+d^{N-j}},\dots,     \alpha ^{md^{N-j+1}-1}     \right)
\nonumber \\\label{gaps}
&=&
\left(\alpha ^{md^{N-j+1}-(d-1)d^{N-j}},\dots,     \alpha ^{md^{N-j+1}-1}   \right)
\eeqn
for $ m=1,2,\dots,d^{j-2}$. The gaps
represent the monomials missing in  $\ket{\widetilde{f}(\alpha)}_{\textbf{A}\setminus A_1 A_j}$ due to the omission of the $j$--th party  in $\ket{f(\alpha)}_{\textbf{A}\setminus A_1}$.

With these denotations we have (omitting subscripts denoting parties)\footnote{It may be of use to note that:
\beq
\bar{g}_{m+1}^j=\alpha^{d^{N-j+1}}\bar{g}_m^j,
\quad 
g_{m+1}^j=\alpha^{d^{N-j+1}}g_m^j
\eeq} :
\beqn
\ket{f(\alpha)}=(1,\alpha,\dots,\alpha^{d^{N-1}-1})&=&( g_1^j,\bar{g}_1^j,\dots, g_{d^{j-2}}^j,\bar{g}_{d^{j-2}}^j)\non
&=& \bigoplus_{m=1}^{d^{j-2}} (g_m^j,\bar{g}_m^j).
\eeqn

We now reshuffle the entries  of $\ket{\psi_1^{(3)}(\alpha)}_{A_1}\otimes \ket{\widetilde{f}_j(\alpha)}_{\textbf{A}\setminus A_1 A_j}$  (we only care about linear independence of the entries so such reordering is allowable) to obtain the following order:
\beqn\label{reshuffled}
&&\hspace{-0.5cm}\ket{\psi_1^{(3)}(\alpha)}_{A_1}\otimes \ket{\widetilde{f}_j(\alpha)}_{\textbf{A}\setminus A_1 A_j}\longrightarrow\non &&(g_1^j,\dots, g_{d^{j-2}}^j\dots, g_m^j P_1(\alpha),\dots, g_m^j P_{d-1}(\alpha),\non  && \hspace{+2cm} g_{m+1}^j P_1(\alpha),\dots, g_{m+1}^j P_{d-1}(\alpha)\dots).
\eeqn
This puts the polynomials $g_x^{j,y} P_z$ ($g_x^{j,y}$ is the $y$--th element of the $x$--th group) in the order of increasing powers of monomials  $\alpha^{\frak{g}} \alpha^{s d^{N-j}}$, $\frak{g}\in g_m$, $s=1,2,\dots,d-1$.

We will now argue that each such polynomial is a sum of monomials of which at least one does not appear in the preceding polynomials and monomials from $g_k^j$, $k=1,2,\dots, d^{j-2}$, in (\ref{reshuffled}). As such, this will prove the required linear independence .
The unique (in the above sense) monomials are actually the ones according to which we have reordered the list of  terms in the above equation. Clearly,   they belong to the gaps introduced in (\ref{gaps}) as they must if the reasoning put forward above is to be applied.

Let $\alpha^{\bar{\frak{g}}}$ be the $(r+1)$--th element of the $m$--th gap, \tzn, 
\beqn\label{brakujacy-rozpis}
&&\bar{\frak{g}}=md^{N-j+1}-(d-1)d^{N-j}+r
\eeqn
with $r=0,1,\dots, (d-1)d^{N-j}-1$.
We will now show that such element can be obtained 
through 
\beqn\label{rozpis-brakujacego}
\alpha^{\bar{\frak{g}}}=\alpha ^{\frak{g}} \alpha ^{s d^{N-j}}
\eeqn
with
\beqn\label{wartosci}
\alpha^{\frak{g}}\in g_m,\quad s=\left\lceil\frac{r+1}{d^{N-j}}\right\rceil.
\eeqn
From (\ref{brakujacy-rozpis}-\ref{rozpis-brakujacego}) we have
\beqn
\frak{g}&=& md ^{N-j+1}-(d-1)d^{N-j} +r -sd^{N-j}     \non
        &=& (m-1)d^{N-j+1} +r -(s-1)d^{N-j},
\eeqn
with $r=0,1,\dots,(d-1)d^{N-j}$.
Let $r=xd^{N-j}+y$, $x=0,1,\dots,d-2$ and $y=0,1,\dots, d^{N-j}-1$, then substituting the value for $s$ from (\ref{wartosci})
\beqn \label{rozpis-g}
\frak{g}=(m-1)d^{N-j+1}+y,
\eeqn
$ y=0,1,\dots, d^{N-j}-1$. which proves the decomposition (\ref{rozpis-brakujacego}). In fact, the given $s$ is unique for $j$ and any $\frak{g}\in g_m$.

Now the claim is that the triple  $(\frak{g},s,j)$ is the unique solution in the meaning introduced above. Let us find other solutions $(\frak{g}',s',j')$
with some $s'\in \{1,2,\dots,d-1\}$, $j' \in \{ 2,\dots, n \}$.
The core of the method is that we only need to care about solutions pointing us to the  polynomials $g_{x}^{y} P_z$, which are to the left [in the sequence (\ref{reshuffled})] of the polynomial under consideration (\tzn, the one for which $\bar{\frak{g}}=\frak{g}+sd^{N-j}$), that is triples such that $\frak{g}'$ belong to a group and 
\beq\label{to-the-left}
\frak{g}'+s'd^{N-j} < \frak{g}+sd^{N-j}.
\eeq

The condition that $(\frak{g}',s',j')$ is another solution giving rise to $\bar{\frak{g}}$ rewrites to:
\beq\label{g-prim}
\frak{g}'=\frak{g}+sd^{N-j}-s' d^{N-j'}.
\eeq
(It is clear that $\frak{g}'=\frak{g} $ iff $s'=s$ and $j'=j$.) With the aid of (\ref{to-the-left}) we can thus narrow our considerations down to the case
\beq \label{obszar-j}
j > j'.
\eeq

Rewriting (\ref{g-prim}) using (\ref{rozpis-g}) we obtain
\beqn
\frak{g}' &=& (m-1)d^{N-j+1}+y     +sd^{N-j}-s' d^{N-j'}\non
&=& (m-s' d^{j-j'-1})d^{N-j+1}-(d-1)d^{N-j}+\non  && \hspace{+3.65cm}+(s-1)d^{N-j}+y, 
\eeqn
where $y=0,1,\dots,d^{N-j}-1$. Taking into account (\ref{obszar-j}) and comparing the above with (\ref{gaps}) we deduce that  $\alpha^{\frak{g}'}$, corresponding to the alleged solution, belongs to a gap [possibly an ,,nonexistent'' one corresponding to $m\le 0$ in (\ref{gaps})]. This is a contradiction with the assumption that it is an element of a group. 

We perform such analysis for all elements from the gaps. We conclude that each $g_{x}^{y} P_z$ in (\ref{reshuffled}) contains a monomial absent in the polynomials to the left of the one under scrutiny. In turn, all elements of (\ref{reshuffled}) are linearly independent  functions. This ends this part of the proof.

Proving the dimensionality of the GES is much less involved. With this aim we need to find the number of linearly dependent polynomials in the entries of (\ref{produkty-nowe-lepsze}). For this count it may be useful to write down explicitly the  vectors from $V_3$
 \beqn
&&\hspace{-1.2cm} (1,\alpha,\dots,\alpha^{d^{N-1}-1},P_1,\alpha P_1,\dots,\alpha^{d^{N-1}-1} P_1,\dots,\non && \hspace{+1.5cm} P_{d-1},\alpha P_{d-1},\dots \alpha^{d^{N-1}-1} P_{d-1}).
 \eeqn
 Since a polynomial $P_k(\alpha)$, $k\ge 1$, is of degree $k d^{N-2}$, the terms, and only those ones, $\alpha ^m  P_{k}(\alpha)$ for $0\le m \le d^{N-1}-d^{N-2}-1$ are linearly dependent on the preceding ones. In turn, there are $(d-1)(d^{N-1}-d^{N-2})$ such linearly dependent terms in (\ref{produkty-nowe-lepsze}). This is exactly the claimed dimension of the GES as there is a total of $d^N$ entries.
\end{proof}

For convenience we have collected the constructions with the dimensions they achieve in Table \ref{tabela-upb}.

\begin{table*}[h!t]
\begin{tabular}{l|l||l}\toprule
    nUPB & $\ket{\psi_1^{(m)}(\alpha)}$ & $\dim \mathrm{GES}$ \\ \hline
    $V_1$ (Theorem \ref{smaller-dim-ges})  & $\big(1,\alpha^{\widetilde{d}},\alpha^{2\widetilde{d}},\dots,\alpha^{(d-1)\widetilde{d}}\big)$, \qquad\qquad
$\widetilde{d}:=\sum_{k=2}^{N-1}(d-1)d^{N-k}+1$   &   $(d-1)^2$  \\
    $V_2$ (Theorem \ref{N-kuditow}) & $\big(1,\alpha^{p_1},\alpha^{p_2},\dots, \alpha ^{p_{d-1}}\big)$, \qquad\qquad  $p_i:= \sum_{k=2}^N i d^{N-k}$  &   $d^N-(2d^{N-1}-1)$\\ 
     $V_3$ (Theorem \ref{produktowe-lepsze}) & $\big (1,P_1(\alpha),P_2(\alpha),\dots, P_{d-1}(\alpha)\big)$,\quad $P_i(\alpha):= \sum_{k=2}^N\alpha^{ i d^{N-k}}$ &   $d^{N-2}(d-1)^2$\\ 
\end{tabular}
\caption{A collection of non--orthogonal unextendible product bases considered in the paper along with the dimensions of the  genuinely entangled subspaces arising from them. In each case the nUPB $V_m$ is given by $V_m=\{\ket{\psi_1^{(m)}(\alpha)}_{A_1}\otimes\ket{\psi_2(\alpha)}_{A_2}\otimes\ldots\otimes\ket{\psi_N(\alpha)}_{A_N}\,|\,\alpha\in\mathbbm{C}\}$ with $\ket{\psi_k(\alpha)}=(1,\alpha ^{d^{N-k}},\alpha ^{2d^{N-k}},\dots,\alpha ^{(d-1)d^{N-k}})$. The nUPBs differ by the form of  $\ket{\psi_1^{(m)}(\alpha)}$. }
\label{tabela-upb}
\end{table*}

Fig. \ref{porownanie-wymiarow-mniejsze-ges} displays the performance of different constructions of GESs as a function of the local dimension $d$ for $N=3$ parties and the number of parties $N$ for $d=3$.

\begin{figure}[h!]
\includegraphics[height=5cm,width=9cm]{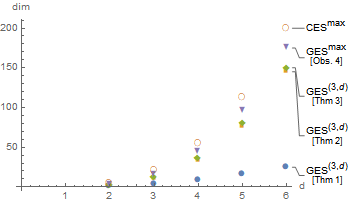}\hfill
\includegraphics[height=5cm,width=9cm]{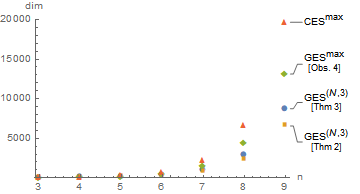}
\caption{Comparison in $d$ of the dimensions of GESs in $(\cee{d})^{\otimes 3}$ constructed with different methods (upper) and in the number of parties $N$ in the qutrit case (bottom). In both plots we have also put a plot of the dimension of a CES  of maximal dimension ($\mathrm{CES}^{\mathrm{max}}$) as a reference.}\label{porownanie-wymiarow-mniejsze-ges}
\end{figure}

It turns out that the choice of the polynomials on $A_1$ is not unique if one wants to obtain a GES but it is optimal as far as the dimension is concerned. We have the following concerning the latter (for the former see Section \ref{generalki}).
\begin{theorem}\label{optymalnosc}
Let $\calV$ be a subspace spanned by the following vectors:
\beqn\label{produkty-optymalne}
\left\{\ket{\varphi(\alpha)}_{A_1}\otimes\ket{\psi_2(\alpha)}_{A_2}\otimes\ldots\otimes\ket{\psi_N(\alpha)}_{A_N}\,|\,\alpha\in\mathbbm{C}\right\}, \non
\eeqn
with $\ket{\psi_k(\alpha)}$, $k=2,\ldots,N$ given by
Eq. (\ref{common}), and
\beq
\ket{\varphi(\alpha)}=\big(Q_0(\alpha),Q_1(\alpha),Q_2(\alpha),\dots,Q_{d-1}(\alpha)\big),
\eeq
where $Q_i(\alpha)$, $i=0,1,\dots,d-1$, are polynomials ordered in the order of the nondecreasing degrees. Then,  if the subspace orthogonal to $\calV$ is a GES, its dimension is no larger than $d^{N-2}(d-1)^2$.
\end{theorem}
\begin{proof}

By assumption, the coordinates of the vectors 
\beqn
&&\hspace{-1cm}\ket{\varphi(\alpha)}_{A_1}\otimes\ket{\psi_3(\alpha)}_{A_3}\otimes\ldots\otimes\ket{\psi_N(\alpha)}_{A_N}=\non
&&\hspace{-0.5cm}(Q_0(\alpha),Q_1(\alpha),Q_2(\alpha),\dots,Q_{d-1}(\alpha))_{A_1}\otimes \non
&& \bigotimes_{k=3}^{N} (1,\alpha ^{d^{N-k}},\alpha ^{2d^{N-k}},\dots,\alpha ^{(d-1)d^{N-k}})_{A_k}=\non 
&&\hspace{-0.5cm}(Q_0(\alpha),Q_1(\alpha),Q_2(\alpha),\dots,Q_{d-1}(\alpha))_{A_1} \otimes \non && (1,\alpha,\alpha^2,\dots,\alpha ^{d^{N-2}-1})_{A_3 A_2\dots A_N}
, \quad \alpha \in \mathbb{C},
\eeqn
which are  polynomials of the form $Q_i \alpha ^{j}$, $i=0,1,\dots,d-1$, $j=0,1,\dots,d^{N-2}-1$, are linearly independent  functions of $\alpha$. Now, in the coordinates of the vectors (\ref{produkty-optymalne}) there appear, among other, polynomials $\alpha^{m d^{N-2}} \alpha^i Q_{d-1}$, $m=1,2,\dots, d-1$, $i=0,1,\dots,d^{N-2}-1$, the degrees of which are larger than the degrees of any of the terms mentioned above. In turn, there are at least $d^{N-2} d+ (d-1)d^{N-2}$  linearly independent polynomials in the entries of (\ref{produkty-optymalne}), which gives an upper bound on the dimension of a GES: $d^N-(2d-1)d^{N-2}=d^{N-2}(d-1)^2$. 
\end{proof}
\subsection{Qubit GES's}\label{qubit-geses}
As an illustration for the construction considered above we give an explicit form of the GESs in the multiqubit case. 

For $d=2$ the vectors constituting the nUPB from Theorem \ref{produktowe-lepsze} are
\beqn\label{wektory-najlepsze-kubity}
&&\ket{0}_{A_1}\sum_{j=0}^{2^{N-1}-1}\alpha^{j}\ket{(j)_2}_{A_2\dots A_N}+\non
&&\hspace{+1.5cm}+\ket{1}_{A_1} \sum_{j=0}^{2^{N-1}-1}\sum_{k=2}^{N}\alpha^{2^{N-k}+j}\ket{(j)_2}_{A_2\dots A_N},
\eeqn
where $(\cdot)_2$ denotes the $(N-1)$ digit binary representation of a number.

By direct computation one can verify that the following set of (unnormalized) non-orthogonal vectors spans the corresponding GES:
\beqn \label{wektory-w-ges}
\ket{0}_{A_1}\sum_{k=2}^{N}\ket{(2^{N-k}+j)_2}_{A_2\dots A_N}&-&\ket{1}_{A_1}\ket{(j)_2}_{A_2\dots A_N},\nonumber\\
\eeqn
with $j=0,1,\dots,2^{N-2}-1$.
For instance, for three parties we obtain
\beqn
\ket{001}+\ket{010}-\ket{100},\quad \ket{010}+\ket{011}-\ket{101}.
\eeqn 
Observe that this $2$--dimensional GES can be completed to a maximal one, \tzn, of dimension $3$ [cf. Eq. (\ref{max-dim-d})] by adding to it a GHZ state with the relative phase changed to $\pi$, that is:
\beqn
\ket{000}-\ket{111}.
\eeqn
That such a subspace is indeed genuinely entangled can be verified in several ways. One is to consider an arbitrary superposition in the subspace and consider all bipartite cuts of the resulting state.
Interestingly, adding the GHZ state with the plus sign [Eq. (\ref{ghz})] does not lead to a GES.

How such completion should be done to achieve the maximal dimension (\ref{max-dim-d}) for any of the constructions given here remains open at this point and needs further treatment. Encouraged by the findings reported above we express the hope that it can be done systematically in an efficient way. 

\subsection{Generalizations}\label{generalki}
We discuss here several generalizations of the presented constructions, which were already announced in the preceding parts of the paper.

First, let us focus on the case of monomial coordinates of the vectors. Assume the space to be $\cee{3}\otimes\cee{3}\otimes\cee{3} $ and consider the following families of vectors ($\alpha \in \mathbb{C}$)
\beqn
(1,\alpha^m,\alpha^n)\otimes (1,\alpha^3,\alpha^6)\otimes (1,\alpha,\alpha^2),
\eeqn
with $m,n \in \mathbb{N}$. By properly varying $m,n$ we can achieve any dimension lower than the optimal one attainable within the monomial construction, which is equal in this case to $10$ (see Theorem \ref{N-kuditow}). For example, for $m=9$ and $n=17$ we obtain a GES composed of a single GME vector.

We now move, also for three qutrits, to the case of different choices of polynomials in Theorem \ref{produktowe-lepsze}.
Let there be given the following families of vectors ($\alpha \in \mathbb{C}$)
\beqn
(1,\alpha+\alpha^p,\alpha^2+\alpha^q)\otimes (1,\alpha^3,\alpha^6)\otimes (1,\alpha,\alpha^2),
\eeqn
with $p,q \in \mathbb{N}$.
By varying $p,q$ we obtain GESs of all dimensions less than or equal to $12$, which is the largest number we can achieve here.

One could also combine both approaches and obtain vectors on $A_1$ with some fraction of the entries being monomials and the rest polynomials. Notably, this would lead in many cases to GESs of the same dimension but different spanning vectors.

Further, notice that the constructions of Theorems \ref{smaller-dim-ges}-\ref{N-kuditow} can be easily generalized to arbitrary, \tzn, not necessarily equal, local dimensions $d_i$.  As an example consider $\calH_{4,3,2}=\cee{4}\otimes \cee{3}\otimes \cee{2}$. The following sets of vectors are nUPBs having GESs in their orthocomplements ($\alpha \in \mathbb{C}$):
\beqn
( 1,\alpha^5,\alpha^{10},\alpha^{15} )_{A_1}& \otimes & ( 1,\alpha^2,\alpha^4 )_{A_2} \otimes ( 1,\alpha )_{A_3},\\
( 1,\alpha^3,\alpha^{6},\alpha^{9} )_{A_1}& \otimes &( 1,\alpha^2,\alpha^4 )_{A_2} \otimes ( 1,\alpha )_{A_3}.
\eeqn
The dimensions of the GESs are, respectively, $3$ and $9$ with the maximal available in this case equal to $11$. 

It is not entirely clear what a generalization of Theorem \ref{produktowe-lepsze} should look like if one were interested in optimizing the dimension, and whether we could benefit at all from considering polynomials  in any case of unequal dimensions. 

In any such case, however, it is possible to apply the same reasoning as above in the case of equal local dimensions to obtain smaller GESs.

\section{Genuinely entangled multipartite states}\label{splatanie-gesy}

Once we have constructed subspaces with the desired properties, it is natural to ask whether they could find an application in a construction of genuinely entangled mixed states, a task which is notoriously difficult. It is quite an obvious conclusion that they can be directly utilized with this purpose.  As a matter of fact, {\it any} state with its support in a GES is GME. Below we discuss a class of such states.

A particularly simple, yet very important, example of a GME state supported on a GES is given by the normalized projection onto it, that is:
\beqn\label{stany-gme}
\rho=\frac{\calP_{\mathrm{GES}}}{d_{\mathrm{GES}}},
\eeqn
where $d_{\mathrm{GES}}$ is the dimension of a GES with projection $\calP_{\mathrm{GES}}$. An interest in such states stems from the fact that their ranks are maximal.

Note that the above makes no assumptions regarding the nature of the subspace orthogonal to a GES. In particular, it could be spanned by an nUPB. Then
\beqn\label{stany-ges}
\rho=\frac{1}{D-d_{\mathrm{nUPB}}}(\jedynka_{D}-\calP_{\mathrm{nUPB}}),
\eeqn
where  $\calP_{\mathrm{nUPB}}$ is a projection onto a $d_{\mathrm{nUPB}}$-dimensional subspace spanned by an nUPB, $D$ and $\mathbbm{1}_D$ are, respectively, the dimension of the whole Hilbert space and the identity operator acting on it.

In this case the construction is an analog of the construction of bound entangled states given first in Ref. \cite{upb-bennett}. Unfortunately, contrary to the case of CESs orthogonal to UPBs, we cannot guarantee that the states (\ref{stany-ges}) are positive after the partial transpose in any cut and in consequence bound entangled (see, \np, Refs. \cite{TothGuhne-fullyPPT,Huber-PPT-BE} for examples of such states). This happens because orthonormalization in general introduces entanglement into the spanning vectors (see Section \ref{preliminary}). The problem with nUPBs in this context was already noticed in \cite{cmp-upb}. In fact, we have evidence that the states (\ref{stany-ges}) are not positive after the partial transpose with respect to any of the cuts \cite{approach}.

Note that the state (\ref{stany-gme}) can be made arbitrarily close to the maximally mixed one (the normalized identity on the whole Hilbert space) by
increasing local dimensions [cf. Eq. (\ref{dim-ges-max})], still being genuinely multiparty entangled.

We conclude this section with an observation that it is straightforward to construct witnesses of genuine entanglement for such states  \cite{LewensteinWitnesses2001}.
\begin{fakt}
The following observable
\beqn
W_{\mathrm{GES}}=\frac{1}{d_{\mathrm{nUPB}}-\epsilon D} (P_{\mathrm{nUPB}}-\epsilon \jedynka_{D})
\eeqn
with 
\beq
\epsilon=\min _{\ket{\psi_{\mathrm{biprod}}}\;}\bra{\psi_{\mathrm{biprod}}}P_{\mathrm{nUPB}}\ket{\psi_{\mathrm{biprod}}},
\eeq
where the minimization is over all pure states that are biproduct states with respect to any bipartition, is a genuine entanglement witness for the states (\ref{stany-ges}).
\end{fakt}

\section{Conclusions and outlook}\label{konkluzje}
We have investigated the relation between unextendible product bases and genuinely entangled subspaces. 
In particular, we have provided ways of constructing  non--orthogonal UPBs leading to subspaces containing solely states being genuinely entangled in their orthocomplement.
Moreover, 
we have also demonstrated how these genuinely entangled subspaces can be used for a construction of genuinely entangled $N$--partite states of any local dimension.  Our study gives further insight into the structure of multipartite entanglement. From a practical point of view, genuinely entangled subspaces are natural candidates for quantum error correction codes, where subspaces with well established properties are utilized  \cite{ScottQECC}. On the other hand, when treated as sources of genuinely entangled multipartite states, they may also find  applications in other areas of quantum information theory, where the usefulness of such states has already been recognized, such as, {\it e.g.}, quantum metrology \cite{GezaMetro2012,HyllusMetro2012,AugusiakMetro2016}, dense coding \cite{PrabhuDenseCoding}, or key distribution \cite{EppingQKD}.

The research presented here provokes several natural questions.
The most obvious is about useful ways of constructing GESs with dimensions not covered by the given constructions, in particular, ones saturating the bound given by (\ref{wymiary-ges}). Phrasing it differently, it is a question about the minimal nUPBs with the desired properties.
Above all, however, it remains unknown whether there exist orthogonal UPBs that give rise
to genuinely entangled subspaces, and if they do exist, which dimensions they achieve. In this respect it would be particularly interesting to verify whether they  exist in all the cases considered here.

From a more general perspective, one could also try to devise methods to build \gesy of the maximal dimension abandoning the requirement for them to be constructed from nUPBs. 
We have only touched upon this problem in the present paper in Sec. \ref{qubit-geses}, where we showed how to complete a three qubit two dimensional GES to a maximal one.
It would  also be interesting to investigate different characteristics, in particular  entanglement properties, of the GESs introduced here. These issues  will be addressed  elsewhere \cite{approach}.

We believe our investigations might give new impetus to the research on non--orthogonal unextendible product bases.
\section{Acknowledgments}
We thank K. \.Zyczkowski for discussions. This project has received funding from the European Union's Horizon 2020 Research
and Innovation Programme under the Marie Sk\l{}odowska-Curie Grant No. 705109.

\bibliography{cytacje3}
\appendix

\section{Turning a set $\calB$ [Eq. \ref{upb-alfa}] into a basis $\bar{\calB}$ [Eq. \ref{baza}] with the spanning property}\label{App-A}

In this appendix we give details of the construction of a basis $\bar{\calB}$ [Eq. \ref{baza}] from a continuous set of vectors  $\calB$ [Eq. \ref{upb-alfa}].
For clarity, we focus on the three-partite setup. The argumentation is unaffected in the case of arbitrary number of parties.

 Let there be given a set of vectors from $\cee{m_1}\otimes \cee{m_2}\otimes \cee{m_3}$:
\beqn
\calB=\{\ket{\Psi(\alpha)}=\ket{\psi_1(\alpha)}_{A_1}\otimes \ket{\psi_2(\alpha)}_{A_2}\otimes \ket{\psi_3(\alpha)}_{A_3} \;  |\;\alpha \in \ce \}, \non
\eeqn
where
\beqn
\ket{\psi_k(\alpha)}_{A_k}=(f_1^{(k)}(\alpha),f_2^{(k)}(\alpha),\dots, f_{m_k}^{(k)}(\alpha))
\eeqn
 with $\{ f_i^{(k)}(\alpha)\}_{i}$, $k=1,2,3$, being linearly independent polynomials. Assume, moreover, the linear independence also to hold for  $\{ f_i^{(k)}(\alpha)f_j^{(l)} (\alpha)\}_{i,j}$, $k,l=1,2,3$, $k\ne l$, being the coordinates of the vectors:
\beqn
\ket{\psi_{23}(\alpha)}_{A_2A_3}:=\ket{\psi_2(\alpha)}_{A_2}\otimes \ket{\psi_3(\alpha)}_{A_3}, \\
\ket{\psi_{13}(\alpha)}_{A_1A_3}:=\ket{\psi_1(\alpha)}_{A_1}\otimes \ket{\psi_3(\alpha)}_{A_3}, \\
\ket{\psi_{12}(\alpha)}_{A_1A_2}:=\ket{\psi_1(\alpha)}_{A_1}\otimes \ket{\psi_2(\alpha)}_{A_2},
\eeqn 
which are local vectors for groups of two parties corresponding to all possible bipartitions, respectively, $A_1|A_2 A_3$, $A_2|A_1 A_3$, and $A_1A_2| A_3$.

 Denote further:
\beqn
u=\dim \spann\; \calB.
\eeqn
Obviously, is is possible to choose $u$ values of $\alpha$ so that the set:
\beq
\bar{\calB}= \{\ket{\Psi(\alpha_i)} \}_{i=1}^u
\eeq
spans the same subspace as $\calB$, in other words, choose a basis for $\spann \; \calB$.

 The claim is now that this choice can always be done in such a way that any $m_k$-- tuple of $\ket{\psi_k(\alpha_i)}$'s spans $\cee{m_k}$ and also any $m_km_l$--tuples of $\ket{\psi_{kl}(\alpha_i)}$'s span corresponding full spaces, {\it i.e.}, $\cee{kl}$. 
If the latter is to hold, any $r$--tuple  of the vectors, with $r$ being smaller than the dimension of the respective full space, must span an $r$--dimensional subspace. Exploiting this observation, we now sketch a simple procedure of building a basis with the desired properties.

Choose an arbitrary value $\alpha=\alpha_1$. We then have the first vector of the basis: 
\beqn
\ket{\Psi(\alpha_1)}&=&\ket{\psi_1(\alpha_1)}\otimes \ket{\psi_2(\alpha_1)} \otimes \ket{\psi_3(\alpha_1)}\non
&=&\ket{\psi_1(\alpha_1)}\otimes \ket{\psi_{23}(\alpha_1)} \non
&=&\ket{\psi_2(\alpha_1)}\otimes \ket{\psi_{13}(\alpha_1)} \non
&=&\ket{\psi_3(\alpha_1)}\otimes \ket{\psi_{12}(\alpha_1)} .
\eeqn
Now, we choose $\alpha=\alpha_2$ in such a way that the functions 
\beqn
\{ \psi_1(\alpha_1), \psi_1(\alpha_2) \}
\eeqn
are linearly independent. It is possible since otherwise it would mean that  
\beq
\{ \ket{\psi_1(\alpha_1)},\ket{\psi_1(\alpha)} \}
\eeq
are linearly dependent functions for all $\alpha$, which obviously is false as the coordinates of $\psi_1$ are linearly independent. The number $\alpha_2$ cannot be arbitrary though. Crucially, however, there is a continuum of good values $\alpha_2$ since the coordinates of $\psi$'s are  polynomials. One explicit way to realize this is by  considering the rank, denote it $r$, of the matrix:
 \beqn
 \calM_1(\alpha):=
\left(
\begin{array}{cccc}
f_1^{(1)}(\alpha_1) & f_2^{(1)}(\alpha_1) & \dots & f_{m_1}^{(1)}(\alpha_1) \\
f_1^{(1)}(\alpha) & f_2^{(1)}(\alpha) & \dots & f_{m_1}^{(1)}(\alpha)
\end{array}
\right) \non
 \eeqn
 with different values of $\alpha$. It holds that $r[\calM_1(\alpha)]=1$ if and only if all $2$ by $2$ minors of $\calM_1$ are rank one. This condition gives a system of polynomial equations on $\alpha$, which only can have a finite number of solutions. Importantly, due to the same reason, $\alpha_2$ can be taken in such a manner that the spanning will also hold for all of the remaining sets of vectors. In such a manner, we construct $\ket{\Psi(\alpha_2)}$ for which each of the following sets spans a two dimensional subspace:
\beqn
 \{\ket{\psi_1(\alpha_1)},\ket{\psi_1(\alpha_2)}\}, \{\ket{\psi_{23}(\alpha_1)},\ket{\psi_{23}(\alpha_2)}\},\\
 \{\ket{\psi_2(\alpha_1)},\ket{\psi_2(\alpha_2)}\},
 \{\ket{\psi_{13}(\alpha_1)},\ket{\psi_{13}(\alpha_2)}\},\\
 \{\ket{\psi_3(\alpha_1)},\ket{\psi_3(\alpha_2)}\},
  \{\ket{\psi_{12}(\alpha_1)},\ket{\psi_{12}(\alpha_2)}\}.
\eeqn
These arguments can now be applied for the remaining values of $\alpha_i$, leading, in turn, to a basis $\bar{\calB}=\{\ket{\Psi(\alpha_i)}\}_{i=1}^u$ with the spanning property across any bipartition as desired.

Let us illustrate this procedure with an elementary relevant example. Consider the following set of vectors:
\beqn\label{przyklad-procedura}
(1,\alpha+\alpha^2)_{A_1} \otimes (1,\alpha^2)_{A_2}\otimes (1,\alpha)_{A_3},\quad \alpha \in \cee{}.
\eeqn
We would like to build a basis corresponding to this set with the spanning property on each subsystem. We begin with taking $\alpha=0$ to obtain $(1,0)\otimes (1,0)\otimes (1,0)$. In the next steps, the values of $\alpha$ can be arbitrary with the only constraint that $\alpha\ne -1$, as otherwise we would not have the spanning property on $A_1$ (or on $A_2$ if we chose $\alpha=1$ at some point).

The set (\ref{przyklad-procedura}) corresponds to the three qubit nUPB considered in Section \ref{polynomial-coordinates}. In the orthocomplement of its span there is a two dimensional GES.
\section{Linear independence of coordinates of vectors on subsystems}\label{App-B}

Here we show the following simple lemma.

\begin{lem}\label{lemat-apendix}
Let there be given a set of linearly independent functions $\{S_i(x) Q_j(x)\}_{i,j}$, $i=1,2,\dots m$, $j=1,2,\dots,n$. Then, also $\{S_i\}_i$ and $\{Q_j\}_j$ are sets of linearly independent functions.
\end{lem}
\begin{proof}
Assume to the contrary that the functions in one of these sets, say $\{S_i\}_i$, are linearly dependent. The latter means that there is a set of $m$ numbers $c_i$, not equal to zero simultaneously, such that 
\beqn\label{na-S}
\sum_{i=1}^{m} c_i S_i(x)=0, \quad \forall x.
\eeqn
 Consider now a linear combination 
 \beqn
\frakL(\{d_{i,j}\}):= \sum_{i,j=1}^{m,n} d_{i,j}S_i(x)Q_j(x).
 \eeqn
Clearly, by taking $d_{i,j}=c_i b_j$ with $c_i$'s such that Eq. (\ref{na-S}) holds and not all $b_k$'s are zero, we obtain:
\beqn
\frakL(\{c_i b_j\})=\sum_{j=1}^{n} b_j \left[\sum_{i=1}^m c_{i}S_i(x)\right] Q_j(x)= 0, \quad \forall x.
\eeqn
This would mean that the functions $S_iQ_j$ are linearly dependent, a contradiction with the assumption of the lemma.  
Hence, $S_i$'s are linearly independent. The same is true for the set of $Q_j$'s.
\end{proof}

The application of this lemma to our purposes is straightforward. We consider $N$--partite product vectors [see Eq. (\ref{upb-alfa})] $\bigotimes_{k=1}^{N}\ket{\psi_k(\alpha)}_{A_k}$ with coordinates being products of functions of $\alpha$ stemming from the multiplication of the coordinates of the local vectors for each of the parties. The latter are linearly independent for any $A_k$ and must be such chosen that the coordinates of any of the vectors of the form 
\beqn
\bigotimes_{k\in I}\ket{\psi_k(\alpha)}_{A_k},\quad I\subset \{1,2,\dots, N\},
\eeqn
are linearly independent. Once we verify that this is the case for a set of $(N-1)$ parties $\textbf{A}\setminus A_j$, we conclude, by Lemma \ref{lemat-apendix}, that it is also true for any subset of  $\textbf{A}\setminus A_j$. Considering all such $(N-1)$--partite systems we arrive at the desired result.

It is important to bear in mind that the implication does not work in the other direction: one can construct sets of linearly independent functions, which after multiplication of their elements give rise to a set of linearly dependent functions. A relevant example [cf. Eq. (\ref{pre-nUPB})] is the set of the symmetric vectors $(1,x,x^2,\dots)\otimes (1,x,x^2,\dots)$ .
\end{document}